\documentclass[num-refs]{wiley-article}
\usepackage[T1]{fontenc}
\usepackage[utf8]{inputenc}

\usepackage{booktabs}
\usepackage{siunitx}
\usepackage[expanded]{mfirstuc}
\newcommand{\charquote}[1]{%
  \ifmmode
\mathopen{\text{`}}\mathtt{#1}\mathclose{\text{'}}%
  \else
    `\texttt{#1}'%
  \fi
}

\DeclareMathOperator{\division}{division}
\DeclareMathOperator{\remainder}{remainder}
\DeclareMathOperator{\floor}{floor}

\DeclareSIUnit{\instruction}{instruction}
\DeclareSIUnit{\float}{float}
\DeclareSIUnit{\characters}{characters}
\DeclareSIUnit{\Mfps}{\mega\float\per\second}
\usepackage{microtype}
\usepackage{lipsum}
\PassOptionsToPackage{hyphens}{url}
\usepackage{hyperref}
\usepackage{multirow}
\usepackage{xcolor}

\newcommand{\ifmakernel}{SIMD algorithm}
\hypersetup{
    colorlinks,
    linkcolor={red!50!black},
    citecolor={blue!50!black},
    urlcolor={black}
}
\usepackage{tikz}
\usetikzlibrary{shapes,arrows,arrows.meta,positioning,fit,calc}

\usepackage{pgfplots}
\usepackage{pgfplotstable}
\usepgfplotslibrary{groupplots}
\pgfplotsset{compat=1.17}
\usepackage{subcaption}
\usepackage{textcomp}

\usepackage{listings}
\definecolor{comments}{HTML}{CA60C2}
\definecolor{keywords}{HTML}{1C21B2}
\definecolor{strings}{HTML}{D02929}
\usepackage{algorithm}
\usepackage[noend]{algpseudocode}

\usepackage[utf8]{inputenc}

\newcommand{\divop}{\mathbin{\mathrm{div}}}

\title{Converting an Integer to a Decimal String in Under Two Nanoseconds}

\author[1]{Jaël Champagne Gareau}
\author[1]{Daniel Lemire}

\corraddress{Daniel Lemire, Université du Québec (TELUQ), Montreal, Quebec, H2S 3L5, Canada}
\corremail{daniel.lemire@teluq.ca}

\fundinginfo{Natural Sciences and Engineering Research Council of Canada, Grant Number: RGPIN-2024-03787\\Fonds de recherche du Québec, \url{https://doi.org/10.69777/361128}}

\affil[1]{Universit\'e du Qu\'ebec (TELUQ), Montreal, Quebec, H2S 3L5, Canada}
\date{}

\begin{document}
\maketitle

\begin{abstract}

Converting binary integers to variable-length decimal strings is a fundamental
operation in computing. Conventional fast approaches rely on recursive division
and small lookup tables. We propose a SIMD-based algorithm that leverages
integer multiply-add
instructions available on recent AMD and Intel processors. Our
method eliminates lookup tables entirely and computes multiple quotients and
remainders in parallel.
Additionally, we introduce a dual-variant design with dynamic
selection that adapts to input characteristics: a branch-heavy variant optimized
for homogeneous digit-length distributions and a branch-light variant for
heterogeneous datasets. Our single-core algorithm consistently outperforms all
competing methods across the full range of integer sizes, running
1.4--2$\times$ faster than the closest competitor and 2--4$\times$ faster than
the C++ standard library function \texttt{std::to\_chars} across tested workloads.

\keywords{Integer numbers, String algorithms, Performance benchmarking}
\end{abstract}

\section{Introduction}

Converting binary integers into their decimal string representations is a
fundamental operation in software systems. It appears in virtually all
applications---whether for displaying values in command line or graphical user
interfaces, writing diagnostic logs, or serializing data to text-based formats
such as JSON, XML, or CSV\@. Despite its ubiquity, this conversion task has
received relatively little attention in the
formal peer-reviewed literature, and common
implementations may leave room for performance improvements. As data formats and
logging systems become performance bottlenecks in modern applications, improving
integer-to-string conversion speed directly benefits real-world software such as
database engines, serialization libraries, and logging frameworks.

Modern processors provide several architectural features that could potentially
be exploited to improve the performance of such conversions. SIMD (Single
Instruction, Multiple Data) extensions, for example, allow the same arithmetic
operation to be applied to multiple data elements in parallel. On recent x86-64
processors from Intel and AMD, the AVX-512 instruction sets can process 512-bit
registers using single instructions, offering potential for accelerating
numerical workloads~\cite{clausecker2023transcoding,mula2020base64}. At the
same time, processors can execute  multiple independent instructions
simultaneously (\emph{instruction-level parallelism}) including load and store
operations. To improve such parallelism, contemporary out-of-order execution
engines can reorder instructions and use speculative execution to hide latencies
and improve throughput. However, speculative execution might suffer from branch
mispredictions. While branch predictors have become increasingly sophisticated,
mispredicted branches  still introduce performance penalties. These
architectural factors interact in subtle ways when implementing high-performance
integer-to-string conversions. Achieving both low latency and high throughput
requires careful coordination between algorithmic structure and processor
microarchitecture, including vector width utilization, instruction scheduling,
and branch prediction behavior.

Our \emph{primary contribution} is a new integer-to-string conversion algorithm
that explicitly leverages SIMD parallelism. Our design exploits the AVX-512
instruction set to process multiple digits concurrently, computing quotients and
remainders in parallel without costly integer divisions. Furthermore, a
distinctive feature of our approach is its dual-variant mechanism, which
dynamically selects between two specialized conversion paths based on sampled
input characteristics: a branch-heavy variant optimized for homogeneous
digit-length distributions and a branch-light variant for heterogeneous
datasets. This adaptive design balances branch prediction accuracy with
instruction-level parallelism.

As a \emph{secondary contribution}, we present, to the best of our knowledge,
the first formal empirical study of integer-to-string conversion performance
on a modern AVX-512-capable processor. Our evaluation combines both randomized
datasets---with controllable digit-length distributions---and data extracted
from real-world files. In particular, we analyze how the distribution of digit
lengths within the input affects throughput, revealing that this factor alone
dominates overall performance. We compare our approach against widely used
library implementations, and our results highlight where each variant excels and
how architectural features influence observed performance.

The remainder of this paper is organized as follows. Section~\ref{sec:related}
reviews prior work and commonly used algorithms, including those found in
standard libraries. Section~\ref{sec:simd} provides a brief overview of SIMD
instructions and the AVX-512 extension relevant to our algorithm.
Section~\ref{sec:math_results} presents the mathematical foundation for our
approach. Section~\ref{sec:methods} describes the design of our proposed
algorithm, with emphasis on its SIMD data layout, parallel digit extraction, and
dynamic variant-selection mechanism. Section~\ref{sec:experiments} describes the
benchmarking methodology, hardware platform, and performance metrics. Finally,
Section~\ref{sec:conclusion} discusses the implications of our results and
outlines possible extensions.

\section{Related Work}%
\label{sec:related}

Formally, the problem consists of converting an integer $n$ into its decimal
representation using the characters $\mathrm{-}$, $\mathrm{0}$, $\ldots$, and
$\mathrm{9}$. In Unicode and ASCII, the digit characters have code point values
ranging from 48 (for $\mathrm{0}$) to 57 (for $\mathrm{9}$) consecutively. Thus,
to convert an integer value $k$ between 0 and 9 to a string, it suffices to
generate the character with the code point value $\mathrm{\charquote{0}} + k$,
where \charquote{0} is a shorthand for the code point value 48. In the general
case, and ignoring negative values, we seek a sequence of characters $c_0 c_1 \ldots c_{k-1}$
such that $n = \sum_{i=0}^{k-1} (c_i - \charquote{0}) \cdot 10^{k-1-i}$, where
$k$ is the number of digits in the representation of $n$.\footnote{The number of
decimal digits needed to represent a positive integer $n$ is $\lfloor \log_{10}
n \rfloor +1$.} In practice, most algorithms first check whether $n$ is
negative; in which case, they prepend a minus sign (\charquote{-}) and proceed
with the absolute value. For simplicity, we focus principally on unsigned
integers and, more precisely, on 64-bit strictly positive values
(i.e., $n \in [1, 2^{64}-1]$).

\begin{figure}[htb]
  \centering
  \begin{tabular}{l}
    \begin{lstlisting}
size_t uint64_to_str_naive(uint64_t n, char *buffer) {
    char *p = buffer;
    while (n > 0) {
        const uint64_t digit = n % 10;
        n /= 10;
        *p++ = '0' + digit;
    }
    std::reverse(buffer, p);
    return p - buffer; // Return length
}
    \end{lstlisting}
  \end{tabular}\restartrowcolors
  \caption{C++ implementation of the classic integer-to-string conversion algorithm~\cite{kernighan1988c}.}%
  \label{fig:naive_listing}
\end{figure}

\paragraph*{Naive and Baseline Methods.}

The classic approach repeatedly divides and takes the modulus by~10 to extract
digits from least to most significant, as illustrated in
Figure~\ref{fig:naive_listing}. Because digits are produced in reverse order, a
final reversal step is required~\cite{kernighan1988c}. Some variants instead write the number
right-to-left into a buffer and then copy the resulting string to its correct position.
More advanced implementations precompute the number of digits
beforehand~\cite{lemire_counting_2025} to place each character directly at its
final position, avoiding both reversal and data movement, thereby enabling a
single-pass implementation.

\paragraph*{Two-Digit Lookup Methods.}

A common optimization is to reduce the number of division and modulus
operations. Instead of extracting one digit per iteration, pairs of digits are
produced by dividing by~100 and mapping the result (0--99) through a small
lookup table of two-character strings. The final odd digit, if any, is handled
separately. This strategy was popularized by
Alexandrescu~\cite{alexandrescu_facebook} and adopted in several
high-performance formatting libraries, such as
\texttt{fmt}.\footnote{\url{https://github.com/fmtlib/fmt/blob/9395ef5fcb81/include/fmt/format.h\#L1219}}
It removes many divisions and relies on cache-friendly, sequential memory
accesses, which are efficiently handled by modern CPUs.


\paragraph*{Division by Constants.}
The conversion of integers to decimal digits often involves divisions by 10 or
powers of two. Division instructions on mainstream computers are typically
slower than multiplications.\footnote{On recent AMD processors with a Zen~5
microarchitecture, a division instruction might require between 11 and 19~cycles
while a multiplication requires only 3~cycles~\cite{Fog2025InstructionTables}.}
For this reason, optimizing compilers replace divisions by \emph{multiplicative
divisions}: a multiplication followed by a
shift~\cite{jacobsohn1973combinatoric,Artzy:1976:FDT:359997.360013,Li:1985:FCD:4135.4990,Magenheimer:1987:IMD:36177.36189,Robison:2005:NUD:1078021.1078059,Granlund:1994:DII:773473.178249,Cavagnino:2008:EAI:1388169.1388172}.
The algorithm used by optimizing compilers usually follows Warren's
approach~\cite{warr:hackers-delight-2e}. We can concisely summarize the
mathematical results with optimal bounds~\cite{lemire2021integer}. Define
$\division(x, y) \coloneqq \floor(x / y)$ and $\remainder(x, y) \coloneqq x -\division(x, y) \cdot y$
for positive real numbers $x, y$ with the constraint that $y\neq 0$. Given a
non-zero divisor $d$ and a non-negative numerator, we have that
$\division(n, d) = \division(c \cdot n, m)$ for all $n \in [0, N]$ if
$1 / d \leq c / m < \left(1 + \frac{1}{N - \remainder(N + 1, d)} \right) 1 / d$,
where $c$ and $m$ are non-negative integers. We can require that $c \in [0, m)$.
By choosing $m$ to be a power of two, the formula $\division(n, d) = \division(c \cdot n, m)$
indicates that we can replace a division by an arbitrary positive integer $d$ by
a multiplication ($c \times d$) followed by a shift (division by $m$).
Fig.~\ref{fig:formula} illustrates how we might compute $c$ in practice.

\begin{figure}[htb]
  \centering
  \begin{tabular}{l}
    \begin{lstlisting}[language=Python]
def find(N, d):
    L = 1
    while True:
        m = 2 ** L
        c = (m + d - 1) // d
        ra = N - (N + 1) % d
        if c * d *  ra < m * (ra + 1):
            return L, c
        L += 1
\end{lstlisting}
  \end{tabular}\restartrowcolors
  \caption{Python script to compute the constant for a multiplicative division.
  It finds $L$ and $c$ such that $\lfloor n / d \rfloor = \lfloor (c \times n) /
2^L \rfloor$ for all $n \in [0,N]$.}%
  \label{fig:formula}
\end{figure}

\paragraph*{Large-Chunk Precomputed Methods.}

The same idea can be extended to larger digit groups, most notably 4-digit and
5-digit blocks. However, table size grows exponentially with block width (40~KB
for a 4-digit table), often exceeding the L1 data cache. Consequently, such
approaches trade arithmetic reduction for cache locality. Examples include
Chateauneu's 4-digit \texttt{itoa} implementation~\cite{chateauneu_itoa_sf2007}
and Henriksen's 5-digit variant~\cite{henriksen2013} derived from the GNU
\texttt{roff} project's \texttt{itoa} implementation. The \texttt{fmt}
contributors later implemented a more aggressive multi-table algorithm
originally proposed by user ``u2985907''~\cite{u2985907_impl}, which heavily
unrolls loops and merges large precomputed chunks to minimize branches and
divisions---trading instruction count for instruction-cache pressure. Hopman's
``\texttt{hopman\_fast}'' algorithm~\cite{hopman2010itoa} uses a similar 4-digit
strategy but attaches a small ``leading-zero'' flag to each chunk, allowing
zero-padding to be skipped quickly when combining ASCII characters. For example,
for the number \num{112345678}, you get the three chunks 0001, 1234, 5678
corresponding to $ n / 10 ^8$, $(n /10 ^4) \bmod 10^4$, and $n \bmod 10^4 $
where $n = 112345678$. You end up with many leading zero bytes in the higher
chunks that must not appear in the final string. Hopman's solution embeds the
number of leading zeroes inside the highest byte of the 4-byte chunk. In
contrast, Henriksen's version operates on 5-digit blocks and maintains separate
forward and reverse tables of strings. The forward tables contain the strings
corresponding to integers 0 to $10^5-1$ inclusively, and from $10^5-1$ to 0
inclusively for the reverse tables. The forward table is used with positive
integers: we lookup the string starting from the beginning of the table. The
reverse table is used for negative integers: starting from the end of the table,
we use a negative index in the table to lookup the string.

\paragraph*{Arithmetic and Fixed-Point Methods.}

Mathisen proposed a fully arithmetic approach that avoids lookup tables
altogether~\cite{mathisen1999fastitoa}. The method partitions a 32-bit integer
into two five-digit segments, processed concurrently by exploiting
instruction-level parallelism. Each sub-conversion uses a 4.28 fixed-point
representation (4 bits integer, 28 bits fractional) in which digits are
extracted via bit shifts and successive multiplications by ten. A small bias
term corrects rounding errors in the reciprocal of 10\,000. By interleaving the
two halves and unrolling the loop, Mathisen's method achieved sub-hundred-cycle
conversion times on late-1990s CPUs without divisions or tables. A later Usenet
post~\cite{mathisen2015convertingNumbersToText} by Mathisen mentioned the
possibility of a 64-bit extension but did not provide a working implementation.

Building on similar fixed-point principles, Hopman's ``fun''
algorithm~\cite{hopman2010itoa} follows a related philosophy:
it recursively splits the number into 8-, 4-, and 2-digit blocks using
reciprocal-multiplication constants approximating division by powers of ten.
Each stage decomposes the value into higher and lower blocks until every block
represents a single digit. Afterward, all digits are converted to ASCII in bulk
via a single \emph{bitwise OR} with \texttt{0x30}, and leading zeros are skipped branchlessly.
This design also exploits instruction-level parallelism and fixed-point
arithmetic to achieve high throughput without tables or divisions.

The AppNexus Common Framework's
implementation\footnote{\url{https://github.com/appnexus/acf/blob/5a4645d/src/an_itoa.c}}
of integer-to-string conversion~\cite{khuong_how_2017} follows a conceptually similar strategy. It
recursively splits the number into base-$10^k$ chunks (with $k \in \{4,8\}$)
using fixed-point reciprocals to compute both quotient and remainder without
divisions.

A related production-grade implementation is Abseil's \texttt{FastIntToBuffer}
routine.\footnote{\url{https://github.com/abseil/abseil-cpp/blob/d97663e/absl/strings/numbers.h\#L204}}
For 32- and 64-bit integers, it decomposes the value into blocks of up to eight
digits, then uses fixed-point reciprocals of $10$, $10^2$, and $10^4$ to compute
decimal digits in parallel inside 32- or 64-bit registers. The resulting bytes
are turned into ASCII by adding packed ``zero'' patterns and storing them with
unaligned 16/32/64-bit writes, so no per-digit or per-block digit table is needed.

\paragraph*{Hybrid Approaches.}

Several recent integer-to-string conversion routines combine elements from both
the lookup-table and arithmetic/fixed-point families, yielding hybrid methods
that avoid costly divisions while keeping table sizes small and cache-friendly.

A representative example is \texttt{jeaiii}'s
algorithm,\footnote{\url{https://github.com/jeaiii/itoa/blob/c861d1c/include/itoa/jeaiii_to_text.h}}
which decomposes the input integer into blocks whose sizes depend on the value
range. Each block is converted using a sequence of reciprocal-multiplication
steps based on fixed-point constants,
followed by two-digit lookups into a compact 200-byte table. The
algorithm aggressively unrolls these operations and relies on a mixture of
masked multiplications, shifts, and lookup pairs to extract multiple digits at once.

A further example of this hybrid style is Yaoyuan's \texttt{itoa\_yy}
implementation.\footnote{\url{https://github.com/ibireme/c_numconv_benchmark/blob/8541662/src/itoa/itoa_yy.c}}
It uses a compact 200-byte lookup table for all two-digit combinations (00--99),
while performing all higher-level decompositions purely via fixed-point
arithmetic. For 32-bit integers, the algorithm selects specialized code paths
for 1--2, 3--4, 5--6, 7--8, and 9--10 digits, each using carefully chosen constants
to implement divisions by $10$, $10^2$, $10^4$, and $10^8$ through
multiplication and bit shifts. The 64-bit path recursively splits the input into
4- and 8-digit blocks, reusing the same two-digit table to materialize all
decimal pairs. There is also a ``large lookup-table'' variant that uses 50 kB
of tables. Both variants are shown to be among the fastest scalar routines
reported in informal benchmarks.\footnote{\url{https://github.com/ibireme/c_numconv_benchmark}}

\paragraph*{SIMD Methods.}

The emergence of SIMD instruction sets has enabled digit extraction to be performed
on multiple integers in parallel. Mu\l{}a~\cite{mula_sse_2011} demonstrated two
SSE2-based routines: the first (originally designed by Wyderski and
adapted by Mu\l{}a) converts two 8-digit 32-bit integers concurrently, while the
second processes a single integer with lower latency. Both use reciprocal
multiplication and vector arithmetic but are limited to 32-bit inputs.

Later experiments by the StackOverflow user
``icecreamsword''~\cite{icecreamsword2015} attempted to vectorize Mathisen's
algorithm using SSE4.1 by splitting numbers into uniform blocks (e.g., three
groups of four digits). The SIMD version achieved roughly a 20--30\% improvement
(34--42~cycles vs.\ 44--50~cycles for the scalar assembly version),
demonstrating modest but consistent benefits from parallelization.

\paragraph*{Implementations in C++ Standard Libraries.}

The C++17 function \texttt{std::to\_chars} provides a standardized low-level
integer-to-string conversion routine whose implementations differ across
libraries and may evolve over time. At the time of writing, the GNU
\texttt{libstdc++} version adopts the same two-digit-table strategy described
above.\footnote{\url{https://gcc.gnu.org/onlinedocs/gcc-15.2.0/libstdc++/api/a00644_source.html\#l00083}}
LLVM's \texttt{libc++} implementation likewise employs a two-digit lookup table
but, for 64-bit values, divides the number by $10^{10}$: it formats the quotient
($<2^{64}/10^{10} < 2^{31}$) when it is non-zero using the 32-bit routine and
then appends the low ten digits 
through an \texttt{append10} helper (which computes ten~digits). This
hierarchical reuse of the smaller routine yields a branch-light and uniform inner
loop.\footnote{\url{https://github.com/llvm/llvm-project/blob/c2b2a347bf94/libcxx/include/__charconv/to_chars_base_10.h}}
By contrast, Microsoft's \texttt{MSVC} implementation retains the classic
repeated-division method shown in
Figure~\ref{fig:naive_listing}.\footnote{\url{https://github.com/microsoft/STL/blob/38d7248/stl/inc/charconv}}
The source code explicitly notes that ``Ry{\=u}'s digit table should be faster
here,'' referring to the lookup table used in the Ry{\=u} float-to-string
conversion algorithm~\cite{ryu2018}, which is essentially the same two-digit
table employed by the two-digit lookup approaches discussed
earlier.\footnote{\url{https://github.com/ulfjack/ryu/blob/023ee6b/ryu/digit_table.h}}

\section{SIMD instructions and AVX-512}%
\label{sec:simd}

Modern processors provide SIMD extensions that apply the same operation to
several data elements packed within wide registers. By contrast with scalar
instructions, which process one value per operation, SIMD instructions exploit
data-level parallelism. Almost all commodity processors today support some form
of SIMD instructions. A SIMD register spanning, say, 64~bytes, can represent
64~bytes, 32~shorts (16-bit words), 16~integers (32-bit words), or 8~long
integers (64-bit words). Similarly, floating-point numbers are supported.

On the x86-64 architecture, SIMD has evolved through several generations. SSE
and SSE2 (128-bit vectors) provided the initial foundation for packed integer
and floating-point arithmetic. AVX (2008) and AVX2 (2013) doubled the register
width to 256~bits and added support for fused multiply-add (FMA) operations as
well as richer integer and gather/scatter instructions. AVX-512, introduced by
Intel starting with the Knights Landing Xeon Phi (2016) and later extended to
mainstream Core and Xeon processors (Skylake-X, Ice Lake, Sapphire Rapids,
Granite Rapids, and equivalents from AMD starting with Zen~4), further widens
the vectors to 512~bits and increases the number of architectural vector
registers from 16 to 32. AVX-512 also introduces an EVEX prefix that enables
operation masking with dedicated mask registers. It becomes possible, for
example, to write, load or multiply only some of the elements of a SIMD
registers. A mask is conceptually just a word made of 8~bits, 16~bits, 32~bits
or 64~bits---corresponding to the number of elements represented in the register.

AVX-512 comprises multiple optional subsets, not all of which are supported on
every processor that implements AVX-512. The foundational AVX-512F subset
provides the core 512-bit floating-point and integer arithmetic. Other important
extensions include AVX-512VL (vector-length orthogonality, allowing most
instructions to operate on 128-bit or 256-bit vectors), AVX-512DQ/BW (32/64-bit
and 8/16-bit integer operations), and AVX-512CD (conflict detection for sparse
gather operations). The subset most relevant to our work is AVX-512IFMA (Integer
Fused Multiply-Add), which adds the instructions \texttt{vpmadd52luq} and
\texttt{vpmadd52huq}. These instructions perform a 52-bit $\times$ 52-bit
unsigned multiplication on each of the eight 64-bit lanes of a 512-bit register,
producing a 104-bit product per lane. The \texttt{vpmadd52luq} variant writes
the low 52~bits of each product (plus an addend) back to the destination, while
the \texttt{vpmadd52huq} variant writes the high 52 bits. When programming, we
can invoke these instructions using \emph{intrinsics} such as
\texttt{\_mm512\_madd52lo\_epu64} and \texttt{\_mm512\_madd52hi\_epu64}. SIMD
intrinsics are special functions giving the programmer access to SIMD
instructions. AVX-512IFMA is available on AMD processors starting with the Zen~4
microarchitecture (2022) and on Intel server processors starting with the Ice
Lake microarchitecture (2019).

The 52-bit precision matches well with decimal computations. For example, to
compute a high-accuracy approximation of $n / 10^k$ or $n \bmod 10^k$ for $k \le 8$,
one can multiply $n$ by a carefully chosen reciprocal constant near $2^{52} / 10^k$
and then extract either the high or low half of the product. The fused nature of
the instruction (multiply + add in one operation) reduces latency and
instruction count compared with separate multiply and add steps.
Furthermore, the \texttt{vpmadd52luq} and \texttt{vpmadd52huq} instructions are
relatively inexpensive: on recent AMD processors with a Zen~5 microarchitecture,
they execute in only four cycles, and two can be retired per
cycle~\cite{Fog2025InstructionTables}. These properties make AVX-512IFMA
particularly attractive for fast decimal digit extraction without lookup tables
or scalar divisions.

\section{Mathematical Result}%
\label{sec:math_results}

Lemire et al.~\cite{lemire_faster_2019} introduced a fast method for computing
remainders using multiplicative inverses. This approach allows us to compute the
remainder of an integer division by a power of ten using only multiplications
and bit shifts We adapt the result to our specific needs and provide a proof of
the following theorem, which is a key component of our algorithm. We state the
main result after recalling the following preliminary lemma from Lemire et
al.~\cite{lemire2021integer}.
Recall that we define
$\division(x, y) \coloneqq \floor(x / y)$ and $\remainder(x, y) \coloneqq x -\division(x, y) \cdot y$
for positive real numbers $x, y$ with $y\neq 0$.

\begin{lemma}\label{lemma:buzz2}
  Consider a positive integer $d > 0$, a non-negative integer $n$, and a
  non-negative real number $x$. We have $\remainder(n,d) = \floor(\remainder(x,d))$
  and $\division(n,d) = \division(x,d)$ if and only if $n \leq x < n + 1$.
\end{lemma}

By substituting a power of two for $m$ and 10 for $d$, the following theorem gives us a formula to compute any digit of the integer $n$ with only two multiplications---when not counting multiplications by a power of two as they can be implemented as a binary shift.

\begin{theorem}\label{theorem:main}
  Consider a base $d \geq 2$, an exponent $k \geq 1$ (e.g., $d = 10$ and $k \in [1, 8]$ for decimal digit extraction),
  an upper bound $N$, and positive integers $c$ and $m$. We have
  \(
    \remainder(\division(n, d^{k-1}), d) = \division(\remainder(c \cdot n + c, m) \cdot d, m)
  \)
  for all integers $n \in [0, N]$ if
  \[
    \left(1 - \frac{1}{N+1}\right) \frac{1}{d^k} \leq \frac{c}{m} < \frac{1}{d^k}.
  \]
\end{theorem}

\begin{proof}
  Let $p = d^{k-1}$ and $D = d^k = d \cdot p$.
  For each $n \in [0, N]$, let $q = \division(n, p) = \floor(n / p)$.
  The goal is to prove
  \[
    \remainder(q, d) = \division\bigl(\remainder(c \cdot n + c, m) \cdot d, m\bigr).
  \]
  We simplify the right-hand side. By definition of $\remainder$,
  \[
    \remainder(c \cdot n + c, m) = c(n+1) - \division(c(n+1), m) \cdot m.
  \]
  Multiplying both sides by $d/m$ gives
  \[
    \remainder(c \cdot n + c, m) \cdot \frac{d}{m} = \frac{c(n+1)d}{m} - \division(c(n+1), m) \cdot d.
  \]
  Let $x = c(n+1)d/m$. The expression above becomes $x - \division(c(n+1), m) \cdot d$.
  Since $x/d = c(n+1)/m$, we have $\division(c(n+1), m) = \floor(x/d)$, and thus $x - \division(c(n+1), m) \cdot d = \remainder(x, d)$.
  Therefore,
  \[
    \division\bigl(\remainder(c \cdot n + c, m) \cdot d, m\bigr) = \floor\bigl(\remainder(x, d)\bigr).
  \]
  Thus we only need to show that $\remainder(q, d) = \floor\bigl(\remainder(x, d)\bigr)$ for every $n \in [0, N]$.
  By Lemma~\ref{lemma:buzz2} (applied with numerator $q$ and real number $x$), the two statements
  \[
    \remainder(q, d) = \floor\bigl(\remainder(x, d)\bigr)
    \qquad\text{and}\qquad
    \division(q, d) = \division(x, d)
  \]
  hold simultaneously if and only if $q \leq x < q+1$.
  In particular, the first equality (the one we need) holds when $q \leq c(n+1)d/m < q+1$, or equivalently
  \[
    \frac{q}{d(n+1)} \leq \frac{c}{m} < \frac{q+1}{d(n+1)}.
  \]
  It therefore suffices to show that the given interval for $c/m$,
  \[
    \left(1-\frac{1}{N+1}\right)\frac{1}{D} \leq \frac{c}{m} < \frac{1}{D},
  \]
  lies inside the required interval for every $n \in [0,N]$. This is equivalent to showing two uniform bounds:

  1. (Upper bound.) For every $n \in [0,N]$, we need to show that
     \[
     \frac{q+1}{d(n+1)} \geq \frac{1}{D}.
     \]
     Equivalently, $(q+1)/(n+1) \geq 1/p$. Write $n = qp + r'$ with $0 \leq r' < p$ (so $q = \floor(n/p)$). Then $n+1 \leq (q+1)p$, and therefore
     \[
     \frac{q+1}{n+1} \geq \frac{q+1}{(q+1)p} = \frac{1}{p},
     \]
     with equality precisely when the block for this $q$ reaches its natural end $n = (q+1)p-1$. (If the last block for $q_{\max} = \floor(N/p)$ is incomplete, the inequality is strict.) Multiplying by $1/d$ yields the claimed lower bound on the required upper bound for $c/m$.

  2. (Lower bound.) For every $n \in [0,N]$, we need to show that
     \[
     \frac{q}{d(n+1)} \leq \left(1-\frac{1}{N+1}\right)\frac{1}{D} = \frac{N}{d p (N+1)}.
     \]
     Equivalently,
     \[
     \frac{q}{n+1} \leq \frac{N}{p(N+1)}.
     \]
     Inside each block where $q$ is constant, the left-hand side is maximized at the smallest $n$ of the block, i.e., at $n = qp$ (where $n+1 = qp+1$). Thus, it is enough to bound the values
     \[
     \frac{q}{qp+1} = \frac{t/p}{t+1} = \frac{t}{p(t+1)}, \quad t = qp.
     \]
     The function $q \to q/(q+1)$ is strictly increasing for $q \geq 0$.
     The global maximum occurs in the last block: let $q_{\max} = \floor(N/p)$, $r = N \bmod p$, so $t = q_{\max}p = N-r$ with $0 \leq r < p$. Then the maximum value attained is
     \[
     \frac{t}{p(t+1)} = \frac{N-r}{p(N-r+1)}.
     \]
     Because $t = N-r \leq N$, we have
     \[
     \frac{t}{t+1} \leq \frac{N}{N+1} \implies \frac{t}{p(t+1)} \leq \frac{N}{p(N+1)}.
     \]
     Consequently every required lower bound on $c/m$ is at most the left endpoint of the given interval.

  Because the given interval for $c/m$ satisfies
  \[
  \max_n \frac{q}{d(n+1)} \leq \left(1-\frac{1}{N+1}\right)\frac{1}{D} \leq \frac{c}{m} < \frac{1}{D} \leq \min_n \frac{q+1}{d(n+1)},
  \]
  the inequality $q \leq x < q+1$ holds for every $n \in [0,N]$. By the one-sided implication of Lemma~\ref{lemma:buzz2} we obtain the desired remainder identity.
  This completes the proof.
\end{proof}

\section{Methods}%
\label{sec:methods}

This section presents our proposed SIMD integer-to-string conversion algorithm.
We begin with an overview of the overall design, followed by detailed
descriptions of the SIMD kernel, multiplicative division technique, and dynamic
variant-selection mechanism.

\subsection{Overview}%
\label{sec:overview}

Our approach decomposes 64-bit integers into two 8-digit chunks by dividing the
integers by $10^8$: $n/10^8$ and $n \bmod 10^8$. Optimizing compilers typically
convert a division by $10^8$ into a multiplication by $\lceil 2^{90}/10^8 \rceil$
followed by a right shift. The remainder can be computed by a multiplication of
the quotient by $10^8$ followed by a subtraction. Our approach then processes
each chunk in parallel using AVX-512 vector instructions to extract all eight
decimal digits simultaneously, and emits the results via masked or fixed-width
stores. The core innovation is a vectorized kernel (\S~\ref{sec:simdalgo}) that
replaces traditional repeated divisions with parallel multiplicative remainders
computed using AVX-512 IFMA instructions. We build two complete conversion
routines on this kernel: a heterogeneous variant (\S~\ref{sec:baseconv})
optimized for mixed digit lengths, and a homogeneous variant
(\S~\ref{sec:homogeneous_specialization}) specialized for uniform-length
batches. A lightweight profiling step (\S~\ref{sec:dynamic_selection})
automatically selects the best variant for each dataset, balancing branch
prediction accuracy and instruction-level parallelism.

\subsection{\xcapitalisewords{\ifmakernel{}}: Parallel Computation of Decimal Remainders and ASCII Packing}%
\label{sec:simdalgo}

The core of our algorithm is a vectorized kernel that converts a single integer
(up to 16 digits) into its ASCII decimal representation using a small number of
SIMD instructions. Rather than extracting digits one or two at a time through
repeated division, we compute all remainders $n \bmod 10^k$ for $k=1,\dots,8$ in
parallel within a single 512-bit register, then convert them to ASCII digits in
bulk. This parallelism is enabled by the Integer-Fused-Multiply-Add (IFMA)
instructions available in the AVX-512 instruction set
extension.\footnote{\url{https://en.wikichip.org/wiki/x86/avx512_ifma}} The IFMA
instructions allow us to compute multiplications and additions on 52-bit limbs
efficiently (each 64-bit lane carries a 52-bit multiplicand), which is ideal for
our purpose of extracting decimal digits.

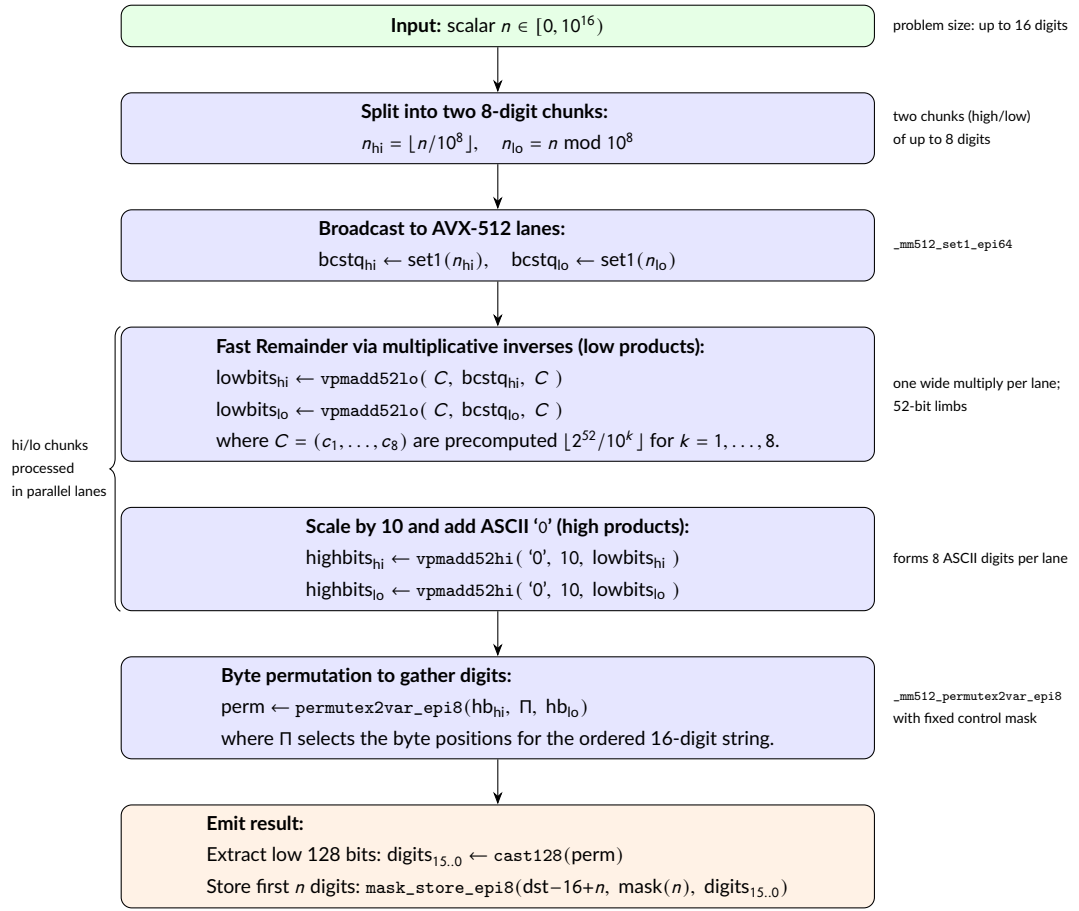
\begin{figure}[htbp]
  \centering
  \begin{tikzpicture}[
    >=Stealth,
    node distance=6mm,
    every node/.style={font=\small},
    box/.style={draw, rounded corners, align=left, inner sep=4pt, minimum width=10cm},
    note/.style={font=\scriptsize, align=left, inner sep=1pt},
    brace/.style={decorate, decoration={brace, amplitude=4pt}},
    boxinput/.style={box, fill=green!10},
    boxop/.style={box, fill=blue!10},
    boxstore/.style={box, fill=orange!10}
  ]
    \node[boxinput] (in) {%
        \textbf{Input:} scalar $n \in [0,10^{16})$%
      };
    \node[note, right=2mm of in] {problem size: up to 16 digits};
    \node[boxop, below=of in] (split) {%
        \textbf{Split into two 8-digit chunks:}\\
        $n_{\text{hi}} = \lfloor n / 10^{8} \rfloor,\quad
        n_{\text{lo}} = n \bmod 10^{8}$%
      };
    \node[note, right=2mm of split] {two chunks (high/low)\\of up to 8 digits};
    \draw[->] (in) -- (split);
    \node[boxop, below=of split] (bcst) {%
        \textbf{Broadcast to AVX-512 lanes:}\\
        $\mathrm{bcstq}_{\text{hi}} \gets \mathrm{set1}(n_{\text{hi}}),\quad
        \mathrm{bcstq}_{\text{lo}} \gets \mathrm{set1}(n_{\text{lo}})$%
      };
    \node[note, right=2mm of bcst] {\texttt{\_mm512\_set1\_epi64}};
    \draw[->] (split) -- (bcst);
    \node[boxop, below=of bcst] (lo) {%
        \textbf{Fast Remainder via multiplicative inverses (low products):}\\
        $\mathrm{lowbits}_{\text{hi}} \gets \texttt{vpmadd52lo}(\;C,\; \mathrm{bcstq}_{\text{hi}},\; C\;)$\\
        $\mathrm{lowbits}_{\text{lo}} \gets \texttt{vpmadd52lo}(\;C,\; \mathrm{bcstq}_{\text{lo}},\; C\;)$\\
        where $C = (c_1,\ldots,c_8)$ are precomputed $\lfloor 2^{52}/10^k\rfloor$ for $k=1,\ldots,8$.
      };
    \node[note, right=2mm of lo] {one wide multiply per lane;\\52-bit limbs};
    \draw[->] (bcst) -- (lo);
    \node[boxop, below=of lo] (hi) {%
        \textbf{Scale by 10 and add ASCII \charquote{0} (high products):}\\
        $\mathrm{highbits}_{\text{hi}} \gets \texttt{vpmadd52hi}(\;\charquote{0},\; \mathbf{10},\; \mathrm{lowbits}_{\text{hi}}\;)$\\
        $\mathrm{highbits}_{\text{lo}} \gets \texttt{vpmadd52hi}(\;\charquote{0},\; \mathbf{10},\; \mathrm{lowbits}_{\text{lo}}\;)$
      };
    \node[note, right=2mm of hi] {forms $8$ ASCII digits per lane};
    \draw[brace,decoration={brace, mirror}] let \p1=(lo.north west), \p2=(hi.south west), \n1={min(\x1,\x2)-0.1} in (\n1,\y1) -- (\n1,\y2) node[midway, left=4pt, anchor=east, note]{hi/lo chunks\\processed\\in parallel lanes};
    \node[boxop, below=of hi] (perm) {%
        \textbf{Byte permutation to gather digits:}\\
        $\mathrm{perm} \gets \texttt{permutex2var\_epi8}(\mathrm{hb}_{\text{hi}},\; \Pi,\; \mathrm{hb}_{\text{lo}})$\\
        where $\Pi$ selects the byte positions for the ordered 16-digit string.
      };
    \node[note, right=2mm of perm] {\texttt{\_mm512\_permutex2var\_epi8}\\with fixed control mask};
    \draw[->] (hi) -- (perm);
    \node[boxstore, below=of perm] (store) {%
        \textbf{Emit result:}\\
        Extract low 128 bits: $\mathrm{digits}_{15..0} \gets \texttt{cast128}(\mathrm{perm})$\\
        Store first $n$ digits: $\texttt{mask\_store\_epi8}(\mathrm{dst}{-}16{+}n,\;\mathrm{mask}(n),\;\mathrm{digits}_{15..0})$
      };
    \draw[->] (perm) -- (store);
  \end{tikzpicture}
  \caption{\ifmakernel{} for integer-to-string split into two 8-digit
  chunks, lane-parallel fastmod via \texttt{vpmadd52}, byte permutation to
assemble ordered digits, and masked store to the output buffer.}
  \label{fig:simd_kernel_visual}
\end{figure}

At its core, the kernel computes all remainders (aka residues) $n \bmod 10^k$
for $k=1, \dots, 8$ \emph{in parallel} and turns them into ASCII digits. In more
detail, our kernel first splits the input integer $n$---in $[0, 10^{16})$---into
two 8-digit chunks (high/low). Each chunk is then broadcast across eight lanes
and multiplied by precomputed constants $c_k = \lceil 2^{52} / 10^k \rceil$ (see
\S~\ref{sec:fastmod}). A second fused multiply-add is then used to scale by $10$
and add the ASCII offset \charquote{0}. Concretely, we use \texttt{VPMADD52LUQ}
to form ``low'' 52-bit products (one per lane), followed by \texttt{VPMADD52HUQ}
to compute $10 \cdot \mathrm{low} + \charquote{0}$. A single byte-wise lane
crossbar (\texttt{permutex2var\_epi8}) then gathers the eight digits of both
8-digit chunks together. Finally, a masked store can write only the significant
prefix to the output buffer when the input $n$ has fewer than 16 digits.
Figure~\ref{fig:simd_kernel_visual} shows the various stages of the computation.

\subsection{Fast Remainder via Multiplicative Inverses}%
\label{sec:fastmod}

\begin{figure}
\centering\small
\begin{subfigure}{0.48\linewidth}
\begin{lstlisting}
__m512i to_string_8digits(uint64_t n) {
  __m512i vn   = _mm512_set1_epi64(n);
  __m512i c = _mm512_setr_epi64(
   45035996, 450359962, 4503599627, 45035996273, 
   450359962737, 4503599627370, 45035996273704, 
   450359962737049);
  __m512i vten    = _mm512_set1_epi64(10);
  __m512i vzero = _mm512_set1_epi64('0');
  __m512i low  = _mm512_madd52lo_epu64(c, vn, c);
  return _mm512_madd52hi_epu64(vzero, vten, low);
}

\end{lstlisting}%
\caption{C function to compute the eight ASCII digits of $n < 10^8$\label{fig:eightdigitscpp}}
\end{subfigure}
\hfill
\begin{subfigure}{0.48\linewidth}
 \begin{lstlisting}
to_string_8digits(unsigned long):
 vmovdqa64     zmm1, zmmword ptr [rip+.L0]
 vpbroadcastq  zmm0, rdi
 vpmadd52luq   zmm1, zmm0, zmm1
 vpbroadcastq  zmm0, qword ptr [rip+.L1]
 vpmadd52huq   zmm0, zmm1, qword ptr[rip+.L2]{1to8}
 ret
\end{lstlisting}%
\caption{Equivalent assembly\label{fig:eightdigitsasm}}
\end{subfigure}
\caption{Side-by-side comparison of a C function to compute eight digits and its assembly equivalent}
\label{fig:eightdigits}
\end{figure}

The \ifmakernel{} from \S~\ref{sec:simdalgo} instantiates the fast computation
of decimal remainders based on Theorem~\ref{theorem:main} to compute the
multiplicative inverses needed for our SIMD kernel. We have that
$\remainder(\division(n, d^{k-1}),d) = \division(\remainder(c \cdot n +c, m) \cdot d, m)$
for all $n \in [0,N]$  if $\left(1-\frac{1}{N+1} \right) 1/d^k \leq c/m < 1/d^k$.
Setting $m = 2^{52}$ and $N = 10^8 - 1$, we can compute the constant $c$ as
$\lfloor 2^{52} / d^k \rfloor$. We can verify that for $d = 10^k$ with $k \in [1, 8]$,
the constant $c$ satisfies the required bounds, ensuring that the formula holds
for all $n \in [0, 10^8)$. See Figure~\ref{fig:mathcheck} for a Python script
that verifies these bounds for the relevant powers of ten. Specifically, we
have that $\remainder(\division(n, 10^{k-1}),10) =\division(\remainder(c \cdot n +c, m) \cdot 10 , m)$
for all $n\in [0,10^8)$ and $k \in [1, 8]$. In practice, we precompute these
constants $\lfloor 2^{52} / d^k \rfloor$ for eight powers of ten and use them in
our SIMD kernel to compute the remainders efficiently. The first step
$\remainder(c \cdot n +c, m)$ is computed using the \texttt{vpmadd52lo}
instruction, which performs the multiplication and addition in one step, while
the second step $\division(\remainder(c \cdot n +c, m) \cdot d , m)$ is computed
using the \texttt{vpmadd52hi} instruction, which scales by $10$ and extracts the
high 52~bits to yield the final digit values.

\begin{figure}
  \centering
  \begin{tabular}{l}
    \begin{lstlisting}[language=Python]
for k in range(1, 9):
  d = 10 ** k
  lhs = (10**8 - 1) * 2**52
  rhs = (2**52 // d) * d * 10**8
  assert lhs <= rhs
\end{lstlisting}
  \end{tabular}\restartrowcolors
  \caption{Python script to verify that the constants $\lfloor 2^{52} / d^k \rfloor$
  satisfy the bounds required by Theorem~\ref{theorem:main} for $d = 10^k$ and $N = 10^8 - 1$.}%
  \label{fig:mathcheck}
\end{figure}

Figure~\ref{fig:eightdigits} illustrates our main routine with a
code snippet and its equivalent assembly. It converts an integer $n < 10^8$ into its eight-digit ASCII string representation using AVX-512 instructions. The C function broadcasts the input $n$ across eight 64-bit lanes in the vector \texttt{vn}. It then utilizes a vector \texttt{c} containing precomputed constants approximating $\lfloor 2^{52} / 10^k \rfloor$ for $k = 8, \ldots, 1$. The computation of \texttt{low} via \texttt{\_mm512\_madd52lo\_epu64(c, vn, c)} adds the original constants to the low 52 bits of the products $n \cdot c$. Finally, \texttt{\_mm512\_madd52hi\_epu64(vzero, vten, low)} adds the ASCII value of `'0'` (broadcast in \texttt{vzero}) to the high 52 bits of $10 \cdot$\texttt{low}, producing the character codes for each digit in the respective lanes of the returned \texttt{\_\_m512i}. The accompanying assembly code is the compiled equivalent, loading the constant vector \texttt{c} into \texttt{zmm1} from memory, broadcasting $n$ into \texttt{zmm0}, computing the low fused multiply-add, then setting \texttt{zmm0} to a broadcast of `'0'` from \texttt{.L1}, and performing the high fused multiply-add with a broadcast of $10$ from \texttt{.L2}, resulting in the digit characters ready for packing or storage. Thus, not counting the return instruction, we require only five instructions to compute the eight-digit ASCII string.

\subsection{Base Conversion Routine}%
\label{sec:baseconv}

Building upon the \ifmakernel{} from \S~\ref{sec:simdalgo}, we
implement the complete integer-to-string conversion in the function
\texttt{avx512\_to\_chars}. This routine converts a 64-bit unsigned integer into
its decimal representation using one or two vectorized 8-digit blocks and a
scalar fallback for the final four digits. It relies on masked vector stores to
handle variable-length numbers without branching.
Figure~\ref{fig:avx_512_to_chars} presents a representative implementation of
this conversion routine in C++.

\begin{figure}[htb]
  \centering
  \begin{tabular}{l}
    \begin{lstlisting}
int avx512_to_chars(uint64_t value, char *const result) {
    const uint32_t n = fast_digit_count(value);

    // --- Main vectorized path for values below 10^16 ---
    if (value < 10000000000000000ULL) {
        const __m128i digits_15_0 = to_string_avx512ifma(value);
        const __mmask16 mask = (__mmask16)(0xFFFFu << (16 - n));
        _mm_mask_storeu_epi8(result - 16 + n, mask, digits_15_0);
        return n;
    }

    // --- Path for larger values: split into 16+4 digits ---
    const auto [q, r] = div10000(value); // q = value / 10^4, r = value % 10^4
    const uint32_t nq = n - 4; // number of digits in q (r has exactly 4 digits)
    const __m128i v16 = to_string_avx512ifma(q); // as illustrated in Fig. 3
    const __mmask16 mask = (__mmask16)(0xFFFFu << (16 - nq));
    _mm_mask_storeu_epi8(result - 16 + nq, mask, v16);
    write_four_digits_10000(result + nq, r);
    return n;
}
    \end{lstlisting}
  \end{tabular}\restartrowcolors
  \caption{Simplified heterogeneous conversion path using the IFMA kernel.
The full implementation also includes a subpath optimized for when the number
has maximum 8-digits that follows the same control logic.}%
  \label{fig:avx_512_to_chars}
\end{figure}

The algorithm first determines the number of digits using a fast branchless
length routine~\cite{lemire_counting_2025}. We first obtain the number of
leading zero bits via \texttt{std::countl\_zero(x)}, which serves as an index
into two precomputed static tables of size 65. The table \texttt{digits} stores
the candidate decimal length corresponding to each possible leading-zero count,
while another table stores the largest integer that still has one fewer digit
than that candidate length. A single comparison \texttt{x > low} then decides
whether the candidate must be incremented, yielding the final digit count.
Because the tables are fully static and the only operations are a leading-zero
count, two array lookups, and a comparison, the whole function executes in a
handful of cycles.

For numbers below $10^{16}$, the conversion fits within two
8-digit blocks and is handled entirely by the IFMA kernel (lines~5--10). For
larger values, the algorithm divides the input by $10^4$, converts the high 16
digits using the same SIMD path, and writes the final four digits with a small
scalar helper (\texttt{write\_four\_digits\_10000}, lines~13--19). The
mask-store operation (line~8~and~17) writes exactly $n$ bytes without
conditional branches, ensuring consistent throughput even when digit lengths
vary. The same control structure is reused for smaller inputs ($<10^8$) using an
8-digit variant of our SIMD kernel, analogous to the 16-digit one and omitted
here for brevity. This implementation serves as the baseline
(\emph{heterogeneous}) variant, balancing performance across heterogeneous
digit-length distributions.

The scalar helper \texttt{write\_four\_digits\_10000} prints the final four
digits by splitting them into two two-digit blocks. Like the table-based methods
described in Section~\ref{sec:related}, it performs a branch-free division
by~100 followed by lookups in a precomputed two-digit table. The standard
approach to this division uses a 64-bit fixed-point multiplier to compute the
quotient $q = (x \cdot \lceil 2^{64} / 100 \rceil) \divop 2^{64}$ and the
remainder $r = x - 100 \times q$. While this yields the exact pair $(q, r)$, it
requires a relatively costly 64×64$\to$128-bit multiplication. Instead, we use
an approximate method that avoids this cost: rather than computing the true
remainder, we derive a \emph{pseudo-remainder} using a 32-bit reciprocal and bit
shifts. This 8-bit pseudo-remainder takes one of 256 distinct possible values,
but each value uniquely corresponds to one of the true remainders in $[0,100)$.
By indexing into a 256-entry lookup table, we recover the correct two-digit
ASCII pair without requiring wide multiplications or subtractions.

\subsection{Homogeneous Specialization}%
\label{sec:homogeneous_specialization}

In practical workloads, integer-to-string conversion is rarely applied to a
single number. Instead, it is typically performed over a large batch of values.
The distribution of digit lengths within such batches can vary widely depending
on the application domain. For instance, in some cases, most numbers may have
the same number of digits (e.g., all eight-digit values), whereas in others the
numbers may span a wide range of lengths.

The previously described heterogeneous variant (\S~\ref{sec:baseconv}) is
designed to handle arbitrary digit-length distributions efficiently. However,
when the input batch is predominantly composed of numbers with the same digit
length, the conversion process can be further optimized by specializing the
routine. Instead of using masked stores to handle variable-length outputs, we
can emit fixed-size stores.

Masked and unmasked vector stores have identical \emph{throughput} on modern
x86~CPUs (e.g., both sustain one store per cycle on the Zen~4
microarchitecture), and although masked stores have higher latency, this
difference is hidden when converting numbers in batch.\footnote{See
\url{https://uops.info/html-instr/VMOVDQU_M128_XMM.html} and
\url{https://uops.info/html-instr/VMOVDQU8_M128_K_XMM.html} for details.}
However, a routine that handles variable-length outputs via masked stores must
also compute the mask at runtime, calculate the destination pointer offset, and
cannot benefit from fixed-size optimizations. These additional instructions add
measurable overhead compared to a routine that can assume a fixed output length
and use direct (unmasked) stores.

Our \emph{homogeneous} variant leverages this observation by assuming most
numbers in the batch share the same digit length. It contains more branches,
each optimized for a specific length. When the dataset is homogeneous, these
branches are consistently well predicted and thus nearly free. Once a path is
chosen, the body is straight-line code using direct (unmasked) stores with
precomputed offsets. Our measurements show that the variable-length
(heterogeneous) variant executes approximately 10--12\% more instructions per
digit, which translates into 30--35\% more cycles due to instruction-level
dependencies. This design suits common cases such as fixed-width identifiers,
Unix timestamps, or telemetry with uniform magnitude.

\begin{figure}[htb]
  \centering
  \begin{tabular}{l}
    \begin{lstlisting}
int avx512_to_chars_homogeneous_17_20(uint64_t value, char *const result) {
  // Split value: q = high digits, r = low 16 digits
  const auto [q, r] = digits::div10e16(value);

  // Write high part (1 to 4 digits) scalarly
  char *p = digits::write_one_two_three_or_four_digits_10000(result, q);

  // Convert remaining 16 digits using SIMD kernel and store directly
  const __m128i digits_15_0 = to_string_avx512ifma(r);
  _mm_storeu_si128(reinterpret_cast<__m128i*>(p), digits_15_0);
  return p - result + 16;  // Total characters = (digits of q) + 16
}
    \end{lstlisting}
  \end{tabular}\restartrowcolors
  \caption{Homogeneous specialization for 17--20-digit inputs. A dispatcher
  selects this path; within it, the 1--4-digit prefix branches resolve to a single
  arm in a homogeneous scenario, yielding excellent branch prediction.}%
  \label{fig:homogeneous_17_20}
\end{figure}

Figure~\ref{fig:homogeneous_17_20} shows the fixed path for 17--20-digit
numbers. A short dispatcher (not shown) branches on the total digit count to
enter this path. Inside the path, the only data-dependent control flow is the
1--4-digit microcase in \texttt{write\_one\_two\_three\_or\_four\_digits\_10000}.
For a fixed total length (17, 18, 19, or 20), the prefix \texttt{q} always has
1, 2, 3, or 4 digits, respectively, so the same arm is taken almost every time
in a homogeneous scenario, yielding excellent branch prediction. The 16-digit
remainder is emitted via a single call to our \ifmakernel{}.

\subsection{Dynamic Variant Selection}%
\label{sec:dynamic_selection}

We implement two specialized variants of our routine: a \emph{homogeneous}
variant (branch-heavy, excellent when most numbers share a fixed digit length)
and a \emph{heterogeneous} variant (branch-light, slightly more
instruction-heavy). The optimal choice depends on the input distribution: logs,
counters, or identifiers often cluster around a few magnitudes (e.g., mostly
seven- or eight-digit values), whereas numeric data extracted from files or
telemetry typically spans a broad range. To adapt automatically, we use a
lightweight \emph{dynamic selection} step. Before converting the full batch, we
sample a small fraction $\alpha$ of elements and estimate the histogram of digit
lengths using the same fast integer-length routine employed by our
\ifmakernel{}. Let $\rho_{\max}$ denote the top-1 (dominant) length ratio in the
sample. If $\rho_{\max}$ exceeds a threshold $\tau_{\text{hom}}$, we select
the homogeneous variant (with the corresponding fixed length); otherwise we
select the heterogeneous variant. Algorithm~\ref{alg:variant-selection}
details the procedure.

\begin{algorithm}[htb]
  \caption{Dynamic selection between homogeneous and heterogeneous variants}
  \label{alg:variant-selection}
  \begin{algorithmic}[1]
    \Require Iterable \textsc{data}; sampling rate $\alpha \in (0,1]$; homogeneity threshold $\tau_{\text{hom}} \in (0,1)$
    \State $m \gets \lceil \alpha \cdot \mathrm{length}(\textsc{data}) \rceil$
    \State $S \gets \mathrm{RandomSample}(\textsc{data}, m)$
    \State Initialize \textsc{counts}$[\ell] \gets 0$ for $\ell = 0 \dots 20$
    \For{$x \in S$}
      \State $\ell \gets \mathrm{digit\_length}(x)$  \Comment{same routine as in \ifmakernel{}}
      \State $\textsc{counts}[\ell] \gets \textsc{counts}[\ell] + 1$
    \EndFor
    \State $c_{\max} \gets \max_{\ell}\ \textsc{counts}[\ell]$
    \State $\rho_{\max} \gets c_{\max} / m$
    \If{$\rho_{\max} \ge \tau_{\text{hom}}$}
      \State \Return \textsc{Homogeneous}
    \Else
      \State \Return \textsc{Heterogeneous}
    \EndIf
  \end{algorithmic}
\end{algorithm}

The sampling rate $\alpha$ and threshold $\tau_{\text{hom}}$ are configurable
parameters. In our experiments, we use $\alpha = 0.01$ (1\%) and
$\tau_{\text{hom}} = 0.95$. The procedure runs in $\Theta(\alpha N)$ time and
$\Theta(1)$ extra memory (the histogram is bounded by the maximum digit length,
e.g., $\le 20$ bins for 64-bit input). Because $\alpha \ll 1$, the overhead is
negligible in practice; quantitative overheads are reported in our evaluation.

\section{Experiments}%
\label{sec:experiments}

We evaluate our proposed integer-to-string conversion algorithm on two main
aspects: (1) the benefits of dynamic variant selection
(\S~\ref{sec:dynamic_selection}), and (2) overall performance compared to the
most relevant existing methods. We begin by describing our experimental setup,
including hardware, datasets, and compared algorithms. We then present our
results and analysis.

\subsection{Systems}%

We conduct our experiments on a machine equipped with an AMD Ryzen~9900X CPU
(Zen~5 microarchitecture). The CPU runs at a base frequency of 4.4~GHz and the
machine has DDR5 (6000\,MT/s) memory. We compile our code using
\texttt{clang++}~21.1.6 with \texttt{-O3 -march=native} optimizations enabled.

We also ran all benchmarks with GCC (\texttt{g++}~15.2.1); while the relative ranking
of algorithms remains similar, absolute performance differs due to differences
in code generation between compilers. The complete \texttt{g++} results are provided online in
the supplementary materials.

\subsection{Data}%

We use both real-world datasets and synthetic numbers to evaluate the algorithms.
For the real-world datasets, We extract integer data from files commonly used
for benchmarking. These files are:
\begin{itemize}
  \item \emph{Twitter}: an export of data from Twitter's API, \num{2108}~numbers;
  \item \emph{CITM catalog}: a catalog of events that occurred in a venue part of CitM (Cité de la Musique) in Paris, 14\,392 numbers;
  \item \emph{StackOverflow}: a dump of Unix timestamps from StackOverflow posts, \num{17823525}~numbers;
  \item \emph{US patents}: patent IDs from  a US patent citation network; represents patents granted between 1975--1999, \num{16518948}~numbers.
\end{itemize}
Figure~\ref{fig:datasets-histograms} shows the distribution of integer lengths
in the first two datasets. We can see that the CITM catalog is more homogeneous,
with 92\% of the integers having 9 decimal digits. In contrast, the Twitter data
exhibits a more heterogeneous distribution, with significant proportions of
small integers (principaly 1--3 digits) as well as larger integers (10 and 18
digits). The other two datasets (StackOverflow timestamps and US patents) are
purely homogeneous, with respectively 10-digit and 7-digit integers only.

\begin{figure}
  \centering
  \begin{tikzpicture}
    \begin{groupplot}[
        group style={
          group size=2 by 1,
          horizontal sep=1.5cm
        },
        ybar,
        every axis plot/.append style={/pgf/bar width=3pt},
        width=0.49\linewidth,
        height=4cm,
        ymin=0,
        ymax=100,
        xticklabel style={font=\small},
        xlabel={Number of decimal digits},
        ymajorgrids
      ]

      \nextgroupplot[
        title={CITM catalog},
        ylabel={Percentage of integers (\%)},
        xtick={5,6,9,13}
      ]
      \addplot+[fill=blue!50, draw=blue] coordinates {
          (5,5.92)
          (6,0.38)
          (9,92.02)
          (13,1.69)
        };

      \nextgroupplot[
        title={Twitter},
        xtick={1,2,3,4,5,6,9,10,18,20},
      ]
      \addplot+[fill=red!50, draw=red] coordinates {
          (1,28.37)
          (2,17.12)
          (3,23.82)
          (4,4.98)
          (5,2.80)
          (6,0.33)
          (8,0.47)
          (9,1.95)
          (10,10.67)
          (18,9.35)
          (20,0.14)
        };
    \end{groupplot}
  \end{tikzpicture}
  \caption{Distribution of integer lengths, normalized to percentages. Left: CITM catalog. Right: Twitter.}
  \label{fig:datasets-histograms}
\end{figure}
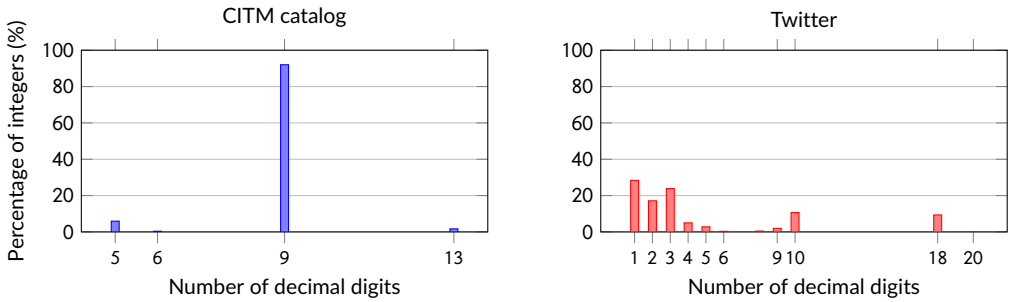

Synthetic datasets are generated to evaluate algorithmic behavior under
controlled digit-length distributions. We consider three synthetic scenarios:
\begin{itemize}
  \item \emph{uniform}: integers whose decimal digit lengths are sampled
    uniformly from 1 to 20;
  \item \emph{natural-8}: integers whose digit lengths are distributed over
    $[1,8]$ and strongly concentrated at 8 digits, with a strictly decreasing
    tail toward shorter representations;
  \item \emph{natural-16}: the same distributional scheme over $[1,16]$.
\end{itemize}
The two \emph{natural} datasets are strongly concentrated at the upper bound. In practice,
approximately 96.8\% of the generated integers have the maximum digit length.
These datasets are designed to mimic the digit-length distribution obtained when
sampling integers uniformly from a large interval $[1,10^k)$: most sampled
values have close to the maximum number of digits, while shorter representations
occur with exponentially decreasing frequency.

\subsection{Software Implementations}%

We benchmark a selection of C and C++ libraries converting 64-bit integers to
their decimal string representation. Our benchmarking code, synthetic data
generators, and datasets are all publicly available
online.\footnote{\url{https://github.com/fastfloat/int_serialization_benchmark}}
The benchmarked libraries and algorithms are the following:
\begin{itemize}
  \item \emph{Naive one-pass}: A baseline implementation based on
    Figure~\ref{fig:naive_listing}, augmented with the one-pass optimization
    where the number of digits is precomputed.
  \item \emph{Champagne--Lemire}: Our algorithm described in
    Section~\ref{sec:methods}. \emph{Unless otherwise specified, we use the version
    employing dynamic variant selection (which automatically chooses between the
    homogeneous and heterogeneous variants).}
  \item \emph{Abseil}: Google's Abseil implementation via the
    \texttt{FastIntToBuffer} function.\footnote{\url{https://github.com/abseil/abseil-cpp.git},
    git hash \texttt{d7aaad8} (April~2024).}
  \item \emph{jeaiii}: The \texttt{to\_text\_from\_integer} function from
    \texttt{jeaiii}'s \texttt{itoa} library.\footnote{\url{https://github.com/jeaiii/itoa/blob/c861d1c/include/itoa/jeaiii_to_text.h},
    git hash \texttt{c861d1c} (November~2022).}
  \item \emph{yy\_itoa}: The \texttt{itoa\_u64\_yy} function by Yaoyuan Guo.\footnote{\url{https://github.com/ibireme/c_numconv_benchmark/blob/8541662/src/itoa/itoa_yy.c},
    git hash \texttt{8541662} (August~2024).}
  \item \emph{AppNexus}: The \texttt{itoa\_u64\_an} function from the AppNexus
    Common Framework library.\footnote{\url{https://github.com/appnexus/acf/blob/5a4645d/src/an_itoa.c},
    git hash \texttt{5a4645d} (December~2017).}
  \item \emph{Hopman}: An implementation of the \texttt{hopman\_fast}
    algorithm.\footnote{\url{https://stackoverflow.com/a/4364057}} We reimplemented
    the algorithm to support 64-bit integers (the original only supports 32-bit
    inputs) and to match the interface used by our other implementations
    (accepting a buffer pointer and returning the output length).
  \item \emph{Mathisen}: An SSE4.1-based algorithm inspired by Mathisen's
    approach, based on the ``four groups of three digits'' variant described by
    the StackOverflow user \texttt{icecreamsword}. Our implementation decomposes
    the input into base-$10^9$ chunks, prints the most significant chunk without
    leading zeros, and formats the remaining chunks as fixed-width blocks using
    SIMD digit extraction. We could not include Mathisen's original code, as it
    does not appear to be publicly available and only supports 32-bit integers.
  \item \texttt{std::to\_chars}: The C++17 standard library integer-to-string conversion function.
\end{itemize}
All of these algorithms are described in more details in Section~\ref{sec:related}.

We exclude Mu\l{}a's SSE-based \texttt{utoa64\_sse}
algorithm~\cite{mula_sse_2011}\footnote{\url{https://github.com/WojciechMula/toys/blob/7731566/sse-utoa/sse64-intrin.c}}
from our benchmarks because its interface differs from the others: it may write
digits at arbitrary positions within the output buffer rather than starting at
the beginning. Adapting it to match our interface would require either copying
the result to the buffer start (underestimating performance) or using AVX-512
masked stores (making it no longer a pure SSE algorithm).

\subsection{Results -- Dynamic Variant Selection}%

We first assess whether the dual-variant strategy is worthwhile by measuring
(1)~the overhead of the selection mechanism itself and (2)~the performance
gains it unlocks. We focus on two representative synthetic distributions:
\emph{uniform} (highly heterogeneous, 1--20 digits) and \emph{natural-16}
(strongly homogeneous, 96.8\% at 16 digits). These extremes bracket the
spectrum of real-world digit-length distributions.

\subsubsection{Dynamic Selection Overhead}%

Table~\ref{tab:variant_selection_overhead} reports the absolute time spent in
Algorithm~\ref{alg:variant-selection} (sampling 1\% of the input) compared to
the total conversion time for datasets of varying sizes. The rightmost column
shows the ratio between conversion time and selection overhead.

\begin{table}
  \caption{Evaluating the variant-selection overhead compared to the entire
  integer-to-string conversion of homogeneous and heterogeneous datasets (g++).}%
  \label{tab:variant_selection_overhead}
  \centering
  \begin{tabular}{lrrrr}
    \toprule
    Model      & Size          & Variant Selection Time (ns) & Integer-to-string Time (ns) & Ratio   \\
    \midrule
    \rowcolor{black!10}
    uniform    &     100\,000  &      38                     &     479\,000                & 12\,605 \\
    \rowcolor{black!10}
               &  1\,000\,000  &     390                     &  4\,950\,000                & 12\,692 \\
    \rowcolor{black!10}
               & 10\,000\,000  &  3\,962                     & 49\,700\,000                & 12\,544 \\
    \rowcolor{white}
    natural-16 &     100\,000  &     110                     &     169\,000                &  1\,536 \\
    \rowcolor{white}
               &  1\,000\,000  &  1\,137                     &  1\,730\,000                &  1\,521 \\
    \rowcolor{white}
               & 10\,000\,000  & 11\,922                     & 17\,500\,000                &  1\,468 \\
    \bottomrule
  \end{tabular}\restartrowcolors
\end{table}

The overhead ratios range from approximately 1\,500 (natural-16) to 12\,500
(uniform), reflecting the fact that heterogeneous data incurs higher conversion
cost (due to branch mispredictions and variable-length handling) while selection
cost remains similar. In the worst case---homogeneous data, where selection
represents the largest \emph{relative} overhead---it consumes only 0.067\% of
total runtime. Because this ratio is nearly constant across dataset sizes, the
selection step scales linearly with negligible marginal cost. Consequently,
always enabling auto-selection incurs no significant penalty in practice, even
when the choice between variants ultimately has little impact.

\subsubsection{Performance Impact of Variant Choice}%

Having established that selection overhead is negligible, we  quantify the
performance benefit of choosing the right variant.
Table~\ref{tab:homogeneous_vs_heterogeneous} compares the per-integer
conversion time (ns/n) of the heterogeneous (\S~\ref{sec:baseconv}) and
homogeneous (\S~\ref{sec:homogeneous_specialization}) variants across both
synthetic and real-world datasets. The ``Selected'' column indicates which
variant the auto-detection heuristic selects for each dataset.

\begin{table}
  \caption{Performance comparison of homogeneous vs heterogeneous variants
  across different datasets. Times shown in nanoseconds per number (ns/n).
  The \textbf{bold} value indicates the faster variant for each dataset.
  The ``Selected'' column shows which variant was chosen by auto-detection.}%
  \label{tab:homogeneous_vs_heterogeneous}
  \centering
  \begin{tabular}{lrrr}
    \toprule
    Dataset & Homogeneous (ns/n) & Heterogeneous (ns/n) & Selected \\
    \midrule
    Uniform 4-digit (1M) & \textbf{0.73} & 0.75 & Homo \\
    Uniform 8-digit (1M) & 0.9 & \textbf{0.72} & Homo \\
    Uniform 16-digit (1M) & \textbf{0.93} & 1.19 & Homo \\
    Uniform 1-20 digits (1M) & 7.86 & \textbf{4.56} & Hetero \\
    Natural-8 (1M) & 1.08 & \textbf{0.73} & Homo \\
    Natural-16 (1M) & \textbf{1.07} & 1.2 & Homo \\
    CITM Catalog & 1.4 & \textbf{1.23} & Homo \\
    Twitter JSON & 1.03 & \textbf{0.93} & Hetero \\
    StackOverflow Timestamps & 1.28 & \textbf{1.26} & Homo \\
    US Patents & \textbf{0.95} & 0.96 & Homo \\
    \bottomrule
  \end{tabular}\restartrowcolors
\end{table}

On uniform 4-digit data, the homogeneous variant outperforms the heterogeneous
one (0.73~vs.~0.75~ns/n, 3\%). The gap widens for 16-digit data
(0.93~vs.~1.19~ns/n, 28\%), consistent with the reduced instruction count
described in \S~\ref{sec:homogeneous_specialization}. The auto-detection
heuristic correctly selects the homogeneous variant in both cases. Curiously, the 8-digit case
reverses this trend under Clang: the heterogeneous variant is faster
(0.72~vs.~0.90~ns/n). Performance counters reveal a compiler code-generation
effect: Clang's heterogeneous code has fewer branches (2~vs.~5) and more instructions per cycle
(6.6~vs.~5.4), overcoming the single extra instruction. The heuristic
misclassifies this case, incurring a 25\% penalty---and the predominantly
8-digit Natural-8 dataset suffers similarly (48\%). This anomaly is
compiler-specific: under GCC, the homogeneous variant wins (0.75~vs.~0.96~ns/n).
More broadly, GCC produces a slower heterogeneous variant across most datasets
(e.g., 1.03~vs.~0.75~ns/n on 4-digit), making the homogeneous variant
relatively more attractive.

Conversely, the highly mixed uniform 1--20 digit distribution strongly favors
the heterogeneous variant: 4.56~ns/n vs.~7.86~ns/n (72\% speedup). Branch
mispredictions in the homogeneous dispatcher outweigh any store optimization,
confirming that branch-light code is essential when digit lengths vary
unpredictably. The heuristic correctly identifies this case.

Real-world datasets show smaller differences: CITM catalog, Twitter,
StackOverflow, and US patents all exhibit gaps under 14\%. In these
cases, misclassification penalties remain modest (8--10\%).

These results validate the dual-variant design: substantial speedups (up to
72\%) when input characteristics clearly favor one implementation, with minimal
cost when both variants perform similarly. The adaptive strategy ensures robust
performance across diverse workloads without manual tuning.

\subsection{Results -- Performance Comparison of Tested Algorithms}

We now compare the overall performance of our proposed algorithm (with dynamic
variant selection enabled) against the other tested algorithms.

\subsubsection{Heterogeneous and Real-World Datasets}%

\begin{figure}[htp]
  \centering
  \includegraphics[width=1\linewidth]{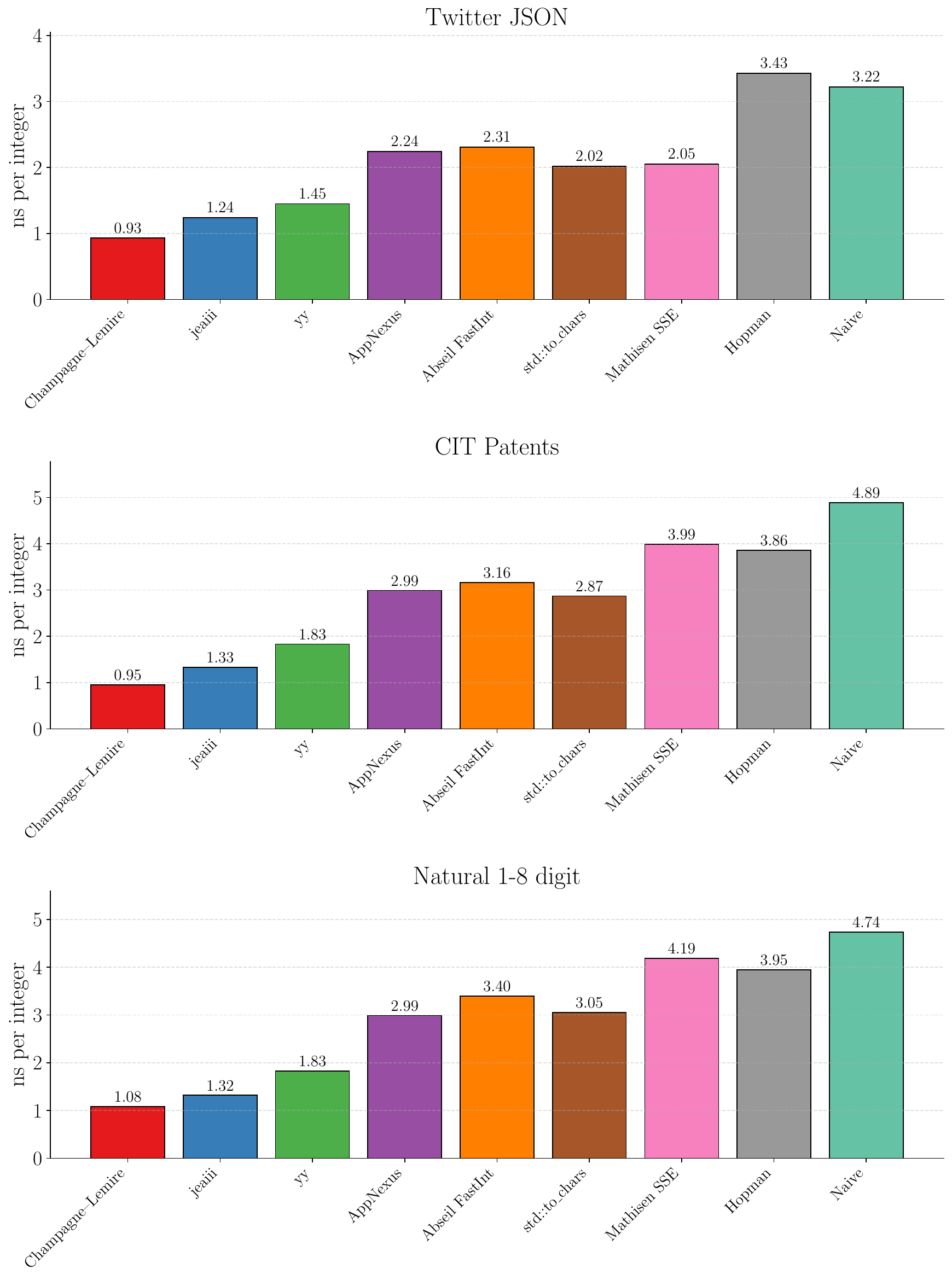}
  \caption{Performance comparison (ns/n, lower is better) of top algorithms
  on three real-world datasets. Our AVX-512 approach outperforms all competitors,
  achieving 1.4~GB/s throughput on Natural~1-8. Extended results with instruction
  and cycle counts are provided in the supplementary materials.}%
  \label{fig:algorithm_comparison}
\end{figure}

Figure~\ref{fig:algorithm_comparison} compares the performance of the most
competitive algorithms across three datasets: Twitter~JSON (variable-length,
69\% with $\le$3 digits), US~patents (homogeneous 7-digit identifiers), and
Natural~1-8 (uniform distribution over 1--8 digits). Our AVX-512 approach
consistently outperforms all competitors. On the Natural~1-8 dataset, we
achieve 0.73~ns/n, corresponding to a throughput of 1.4~GB/s of string
output---nearly twice the rate of the next best competitor (\texttt{jeaiii} at
0.75~GB/s). The performance gap is particularly pronounced on homogeneous
datasets: on US~patents, \texttt{jeaiii} is 39\% slower, and on Natural~1-8,
it is 81\% slower. Even on Twitter---where scalar approaches traditionally
excel due to the prevalence of small integers---our algorithm maintains a 33\%
advantage.

\subsubsection{Homogeneous Fixed Digit-Length Datasets}%

Fixed digit-length benchmarks are common, though they should
be interpreted with care: an algorithm that performs well on each fixed length
individually may still underperform on variable-length workloads due to branch
mispredictions or mode-switching overhead. Nevertheless, homogeneous datasets
do occur in practice---database identifiers, timestamps, and telemetry counters
often cluster around a single magnitude---and our previous experiments on
real-world datasets already demonstrate strong performance under variable-length
conditions. We therefore include fixed-length benchmarks to provide a complete
picture of per-length behavior.

Figure~\ref{fig:homogeneous_sizes} reports results across all digit lengths
(1 to 20 digits). For single-digit numbers, \texttt{std::to\_chars} is fastest
at 0.42~ns/n, compared to 1.25~ns/n for our algorithm. For 2-digit numbers,
\texttt{jeaiii} leads at 0.72~ns/n versus 1.08~ns/n for ours. At 3 digits we
are essentially tied with \texttt{jeaiii} (0.90~vs.~0.91~ns/n), but starting at
4 digits we pull decisively ahead and dominate all competitors across all
remaining sizes. The advantage is
particularly striking at 8 digits, where our algorithm achieves 0.72~ns/n---an
85\% speedup over \texttt{jeaiii} (1.33~ns/n). Beyond 10 digits,
\texttt{jeaiii}'s performance becomes erratic, with notable slowdowns at 10 and
18 digits. These results demonstrate that our algorithm's performance advantage
holds across the full spectrum of digit lengths, not just on specific real-world
distributions.

\begin{figure}[htp]
  \centering
  \includegraphics[width=1\linewidth]{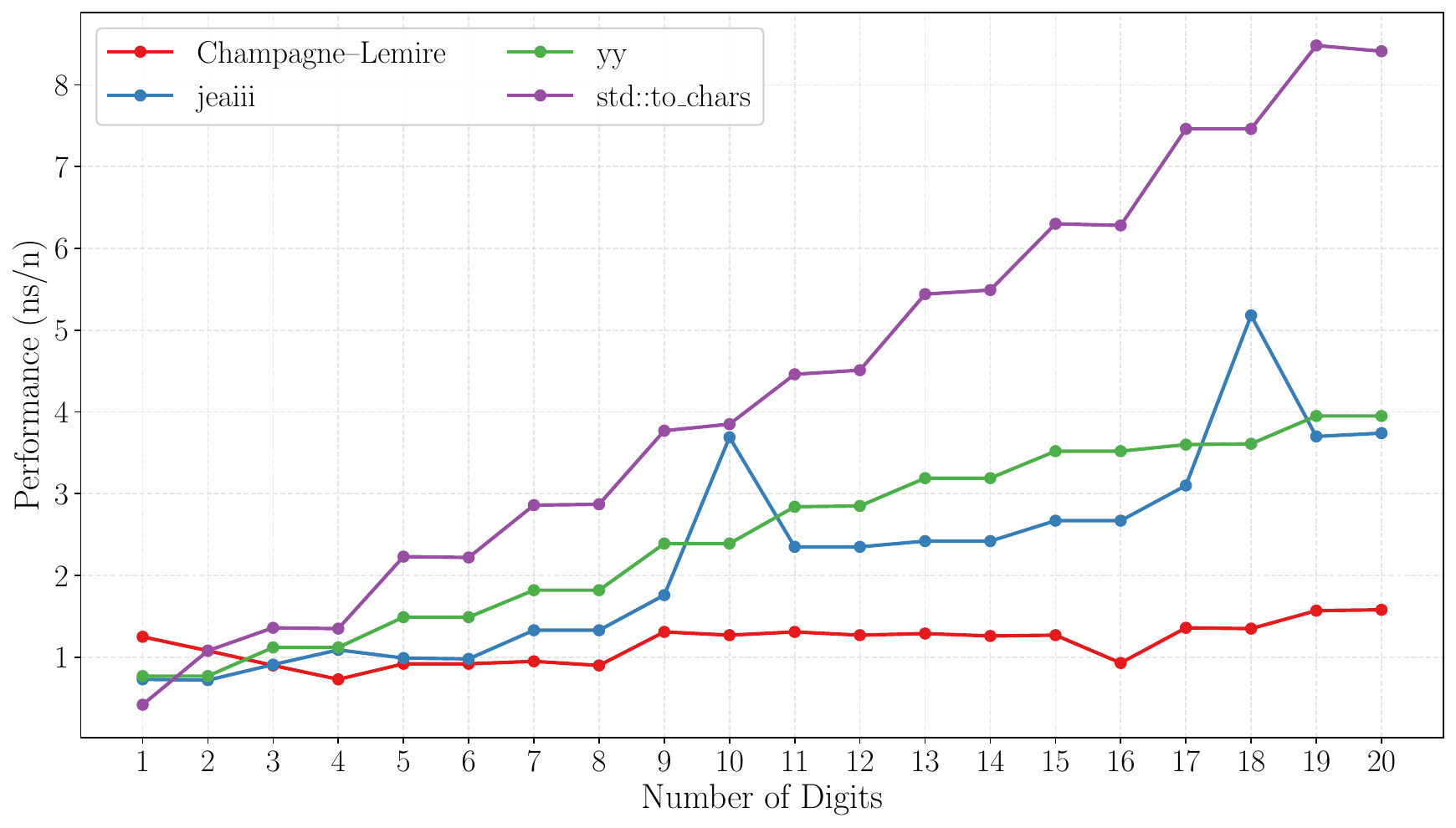}
  \caption{Performance across all digit lengths (1--20) on homogeneous datasets
  of 1M integers each. Our AVX-512 algorithm dominates from 3 digits onward. A
  complete comparison including all algorithms is available in the supplementary
  materials.}%
  \label{fig:homogeneous_sizes}
\end{figure}

\section{Conclusion}%
\label{sec:conclusion}

We presented the first AVX-512 integer-to-string conversion algorithm in the
literature. Our approach leverages AVX-512 IFMA instructions to extract multiple
digits in parallel, achieving high throughput across nearly all digit lengths.
To handle variable-length outputs efficiently, we designed both heterogeneous
(branch-light, masked stores) and homogeneous (branch-heavy, unmasked stores)
variants, along with a dynamic selection mechanism that adaptively chooses the
optimal implementation based on input characteristics. Our experiments
demonstrate that this dual-variant strategy delivers substantial speedups
(up to 72\%) on synthetic datasets with extreme digit-length distributions, while
imposing negligible overhead (0.067\% worst-case) on real-world data where both
variants perform similarly. Overall, our AVX-512 algorithm outperforms all
existing state-of-the-art methods across all tested datasets, running up to
twice as fast as the best competing scalar approach (\texttt{jeaiii}) and up to
four times faster than the standard library's \texttt{std::to\_chars}. This work
highlights the potential of SIMD techniques for accelerating integer-to-string
conversion and provides a robust, high-performance solution across diverse
workloads.

Two natural extensions of this work deserve further investigation. First, while
our implementation targets AVX-512 IFMA instructions, it would be valuable to
evaluate how the proposed SIMD-based approach maps to alternative vector
architectures, in particular ARM platforms supporting SVE or SVE2. Although the
available instruction sets differ, the underlying parallel digit extraction
strategy may transfer with suitable adaptations. Second, our study focuses on
converting one integer at a time; exploring SIMD-oriented approaches that
convert groups of integers in batch could enable a different class of
vectorization strategies, potentially improving throughput in data-parallel
serialization workloads.

Although our study focuses on integer-to-string conversion, the problem is
closely related to floating-point number formatting, which is inherently more
complex. In floating-point conversion, one must first determine the correct
decimal exponent and mantissa from the binary representation before producing
the textual output. The final stage of this process---writing the decimal
mantissa---is conceptually similar to integer-to-string conversion, but over
a smaller range (e.g., up to seventeen digits for IEEE~754 double-precision
values versus up to twenty digits for 64-bit unsigned integers). Additional
challenges include placing the decimal point, supporting both fixed and
scientific notations, and ensuring rounding correctness. While these aspects
require extra handling, the core operation---converting the decimal mantissa to
its textual form---shares the same computational bottlenecks as
integer-to-string conversion. Consequently, advances in integer conversion
techniques can accelerate part of the core of floating-point formatting
routines, with adaptations needed to address their specific requirements.

\section*{Author Contributions}

Jaël Champagne Gareau: conceptualization; investigation; software; experimentation; writing-review and editing.
Daniel Lemire: conceptualization; software; validation; experimentation; data analysis; writing-original draft; writing-review and editing.

\section*{Funding Information}

This work was supported by the Natural Sciences and Engineering Research Council of Canada, Grant Number: RGPIN-2024-03787.
The first author is supported by a postdoctoral grant from Fonds de recherche du Québec, \url{https://doi.org/10.69777/361128}.

\section*{Data Availability Statement}

All our data and software is freely available online. The C++ benchmarking
software, along with the datasets used, is available at
\url{https://github.com/fastfloat/int_serialization_benchmark}. The
supplementary materials referenced throughout this paper, including complete
performance tables and raw benchmark data, are available on the paper's webpage
at \url{https://www.jaelgareau.com/en/publication/gareau_lemire-spe26/}.

\renewcommand{\refname}{References}
\bibliography{ref}

\appendix
\newpage
\section{List of acronyms}

\begin{table}[htbp]
  \begin{tabular}{@{}ll@{}}
    \toprule
    \textbf{Acronym} & \textbf{Full form} \\
    \midrule
    ASCII            & American Standard Code for Information Interchange \\
    AVX              & Advanced Vector Extensions \\
    AVX-512          & Advanced Vector Extensions 512-bit \\
    AVX-512BW        & AVX-512 Byte and Word \\
    AVX-512CD        & AVX-512 Conflict Detection \\
    AVX-512DQ        & AVX-512 Doubleword and Quadword \\
    AVX-512F         & AVX-512 Foundation \\
    AVX-512IFMA      & AVX-512 Integer Fused Multiply-Add \\
    AVX-512VL        & AVX-512 Vector Length \\
    AVX2             & Advanced Vector Extensions 2 \\
    CSV              & Comma-Separated Values \\
    EVEX             & Enhanced Vector Extensions (AVX-512 prefix) \\
    FMA              & Fused Multiply-Add \\
    IFMA             & Integer Fused Multiply-Add \\
    ns               & nanoseconds \\
    ns/n             & nanoseconds per integer \\
    SIMD             & Single Instruction, Multiple Data \\
    SSE              & Streaming SIMD Extensions \\
    SSE2             & Streaming SIMD Extensions 2 \\
    SSE4.1           & Streaming SIMD Extensions 4.1 \\
    \bottomrule
  \end{tabular}
\end{table}

\end{document}